\documentclass{article}

\usepackage[margin=1in]{geometry}

\usepackage{amsmath,amsfonts,amsthm}
\usepackage{xcolor}
\usepackage{amssymb}
\usepackage{xspace}
\usepackage{hyperref}

\usepackage{thmtools,thm-restate} 
\usepackage{mathtools}
\usepackage{graphicx,subcaption}
\usepackage{algorithm,algpseudocode}

\newtheorem{theorem}{Theorem}
\newtheorem{lemma}{Lemma}
\newtheorem{claim}{Claim}
\newtheorem{corollary}{Corollary}
\newtheorem{definition}{Definition}

\def\E{\mathbb{E}}
\def\R{\mathbb{R}}
\def\Z{\mathbb{Z}}

\DeclarePairedDelimiter{\abs}{\lvert}{\rvert}

\renewcommand{\Pr}{\mathbb{P}}

\newcommand{\polylog}{\mathrm{polylog}}

\newcommand{\eps}{\varepsilon}

\newcommand{\cost}{\textsc{cost}}

\newcommand{\cell}{C}
\newcommand{\grid}{\mathcal{C}}
\newcommand{\subgrid}{\mathcal{\widetilde{C}}}
\newcommand{\subcell}{\widetilde{C}}

\newenvironment{subproof}[1][\proofname]{%
  \begin{proof}[#1]%
}{%
  \end{proof}%
}

\newcommand{\NOCOMMENT}{}

\ifdefined\NOCOMMENT
    \newcommand{\snote}[1] {}
    \newcommand{\anote}[1] {}
    \newcommand{\dnote}[1] {}

\else
\newcounter{note}[section]
\renewcommand{\thenote}{\thesection.\arabic{note}}
\newcommand{\snote}[1]{\refstepcounter{note}\textcolor{brown}{$\ll${\bf Sasho~\thenote:} {\sf #1}$\gg$\marginpar{\tiny\bf SN~\thenote}}}
\newcommand{\anote}[1]{\refstepcounter{note}\textcolor{magenta}{$\ll${\bf Andrey~\thenote:} {\sf #1}$\gg$\marginpar{\tiny\bf ABK~\thenote}}}
\newcommand{\dnote}[1]{\refstepcounter{note}\textcolor{cyan}{$\ll${\bf Dmitry~\thenote:} {\sf #1}$\gg$\marginpar{\tiny\bf DP~\thenote}}}
\fi

\newcommand{\cut}[1]{}

\newmuskip\pFqmuskip

\title{Preconditioning for the Geometric Transportation Problem}

\author{Andrey Boris Khesin\\
  University of Toronto\\
  \url{andrey.khesin@mail.utoronto.ca}
  \and
  Aleksandar Nikolov\thanks{Supported by an NSERC Discovery Grant (RGPIN-2016-06333). Part of
    this work was done while this author was visiting the Simons Institute for the
    Theory of Computing semester on Bridging Continuous and Discrete
    Optimization, partially supported through NSF grant \#CCF-1740425}\\
  University of Toronto\\
  \url{anikolov@cs.toronto.edu}
  \and
  Dmitry Paramonov\\
  University of Toronto\\
  \url{dmitry.paramonov@mail.utoronto.ca}
}

\date{}

\begin{document}

\maketitle
\begin{abstract}
  In the geometric transportation problem, we are given a collection of points $P$ in $d$-dimensional Euclidean space, and each point is given a supply of $\mu(p)$ units of mass, where $\mu(p)$ could be a positive or a negative integer, and the total sum of the supplies is $0$. The goal is to find a flow (called a transportation map) that transports $\mu(p)$ units from any point $p$ with $\mu(p) > 0$, and transports $-\mu(p)$ units into any point $p$ with $\mu(p) < 0$. Moreover, the flow should minimize the total distance traveled by the transported mass. The optimal value is known as the transportation cost, or the Earth Mover's Distance (from the points with positive supply to those with negative supply). This problem has been widely studied in many fields of computer science: from theoretical work in computational geometry, to applications in computer vision, graphics, and machine learning.

  In this work we study approximation algorithms for the geometric transportation problem. We give an algorithm which, for any fixed dimension $d$, finds a $(1+\eps)$-approximate transportation map in time nearly-linear in $n$, and polynomial in $\eps^{-1}$ and in the logarithm of  the total supply. This is the first approximation scheme for the problem  whose running time depends on $n$ as $n\cdot \mathrm{polylog}(n)$. Our techniques combine the generalized preconditioning framework of Sherman, which is grounded in continuous optimization, with simple geometric arguments to first reduce the problem to a minimum cost flow problem on a sparse graph, and then to design a good preconditioner for this latter problem.
\end{abstract}

\section{Introduction}




In the \emph{Geometric Transportation problem}, we are given a set $P$ of $n$
points in $d$-dimensional Euclidean space, and a function $\mu:P \to
\Z$ that assigns a (positive or negative integer) weight to each
point, so that $\sum_{p \in P}{\mu(p)} = 0$. We call $\mu$ the
supply function.
We can think of each point $p \in P$ as either a pile of earth, or a
hole, and $\mu(p)$ gives, respectively, the amount of earth (if positive) or the
size of the hole (if negative). The constraint on $\mu$ means that the total space
in the holes equals the total amount of earth in the piles. Our goal
is to find the most efficient way to transport the earth to the
holes, where the cost of transporting a unit of mass from $p$ to $q$ equals the
distance $\|p - q\|_2$ it must travel, measured in the Euclidean
norm. More formally, 
 let $P_+ = \{p \in P: \mu(p) \ge 0\}$ be the
 ``piles'' and $P_- = \{q \in P: \mu(q) < 0\}$ be the ``holes''. We
want to solve the minimum cost flow problem given by the following
linear program
\begin{align*}
  \text{Minimize } &\sum_{p \in P_+, q \in P_-} f_{pq} \|p - q\|_2\\
  \text{subject to}&\\
  \forall p \in P_+: &\sum_{q \in P_-}{f_{pq}} = \mu(p),\\
  \forall q \in P_-: &\sum_{p \in P_+}{f_{pq}} = \mu(q),\\
  \forall p \in P_+, q\in P_-: &f_{pq} \ge 0.
\end{align*}
Above, the constraints enforce that all demands are satisfied,
i.e.~all the earth goes into holes, and no hole is overfilled. A
vector $f \in \R^{P_+ \times P_-}$ that satisfies these constraints is
called a \emph{transportation map} for $(P, \mu)$. 
The cost of an optimal transportation map is called the
\emph{transportation cost} of $(P, \mu)$, or also the Earth Mover's
Distance, or $1$-Wasserstein distance between the measures $\phi = \mu|_{P_+}$ and $\psi =
-\mu|_{P_-}$, and is denoted $\cost(P,\mu)$.

The geometric transportation problem is a discretized version of the
continuous optimal mass transportation problem of Monge and
Kantorovich. This is a classical problem in mathematics, with many
beautiful connections to partial differential equations, geometry, and
probability theory, among others~\cite{Villani-Topics}. The discrete
problem above has also found a large number of applications, for
example to shape matching~\cite{GiannV02,GraumanD04} and image
retrieval~\cite{RubnerTG00,QinCL04} in computer vision, and interpolation
between distributions in graphics~\cite{Bonneel11}. It is also of
intrinsic interest as a natural problem in computational geometry.

The geometric transportation problem can be modeled as a minimum cost
flow problem in the complete bipartite graph with bi-partition $P_+
\cup P_-$. This graph is dense, and solving the minimum cost flow
problem exactly using the algorithm of Lee and Sidford would take time
$O(n^{2.5}\ \polylog(U))$,\footnote{We allow the implicit constants in
  the asymptotic notation to depend on the dimension $d$, except when
  the asymptotic notation is used in an exponent.} where $U =
\sum_{p \in P_+}\mu(p)$. If instead we settle for an
$(1+\eps)$-approximation, then the recent algorithm of
Sherman~\cite{Sherman17} gives running time $O(n^{2+o(1)}\eps^{-2})$.
However, there exist faster (in certain parameter regimes)
approximation algorithms which exploit the geometric structure of the
problem: after a long line of work, Sharathkumar and
Agarwal~\cite{SharathkumarA12-SODA} gave an algorithm that computes a
$(1+\eps)$-approximation in time $O(n\sqrt{U}\ \polylog(U,\eps,n))$,
and, recently, Agarwal et al.~\cite{AgarwalFPVX17} gave several
algorithms with different trade-offs, among them an
$(1+\eps)$-approximation algorithm running in time
$O(n^{3/2}\eps^{-d}\ \polylog(U,n))$. A nearly linear time
$(1+\eps)$-approximation algorithm was given by Sharathkumar and
Agarwal~\cite{SharathkumarA12-STOC} in the special case of unit
supplies; their algorithm runs in time $O(n\eps^{-O(d)}\
\polylog(n))$. Until now, no such algorithm was known for general
supply functions. 

A related line of work focuses on estimating the transportation cost,
without necessarily computing a transportation map. An influential
paper in this direction was that of Indyk~\cite{Indyk07}, whose
algorithm gives a constant factor approximation to the transportation
cost in the case of unit supplies, in time $O(n\ \polylog(n))$. This
result was extended to arbitrary supplies and to approximation factor
$1+\eps$ by Andoni et al.~\cite{AndoniNOY14}, whose algorithm runs in
time $O(n^{1+o(1)})$, with the $o(1)$ factor in the exponent hiding
dependence on $\eps$.  The result of Sherman~\cite{Sherman17},
mentioned above, together with the existence of sparse Euclidean
spanners, implies an $O(n^{1+o(1)}\eps^{-O(d)})$ time algorithm to
estimate the transportation cost.\footnote{We are indebted to an
  anonymous reviewer for this observation.} Indeed, using the
well-separated pair decompositions of Callahan and
Kosaraju~\cite{CallahanK93} (see also Chapter 3 of
\cite{Sariel-book}), one can construct in time $O(n\log n +
n\eps^{-d})$ a graph $G$ on $P$ with $O(n\eps^{-d})$ edges, so that
the shortest path in $G$ between any $p,q\in P$ has total length at
most $(1+\eps)\|p-q\|_2$. It is then simple to construct a minimum
cost flow problem on $G$ whose optimal value equals the transportation
cost of $(P, \mu)$ up to a factor of $1+\eps$. This minimum cost flow
problem in $G$ can then be solved with Sherman's algorithm to achieve
the claimed running time. It is not immediately clear, however, how to
use the flow in $G$ to construct a transportation map of comparable
cost in nearly linear time. 

There is a related line of
work~\cite{PhillipsA06,AgarwalS14,Altschuler-W2} that studies the
transportation problem when the cost of transporting mass from $p$ to
$q$ is $\|p-q\|_2^r$, giving the $r$-Wasserstein distance. This
appears to be a more challenging problem, and we do not address it
further.

\subsection{Our Results and Methods}

Let us recall that
the aspect ratio (or spread) of a pointset $P$ is defined as the ratio of its
diameter to the smallest distance between two distinct points.
Our main result is a nearly linear time $(1+\eps)$-approximation
algorithm for the geometric transportation problem, as captured by the
following theorem.
\begin{theorem}\label{thm:main}
  There exists a randomized algorithm that on input an $n$-point set
  $P\subset \mathbb{Q}^d$ with aspect ratio $\Delta$, and supply
  function $\mu:P \to \mathbb{Z}$, runs in  time $O(n\eps^{-O(d)}
  \log(\Delta)^{O(d)}\log n)$, and with probability at least $1/2$
  finds a transportation map with cost at most $(1+\eps)~\cdot~\cost(P,\mu)$.

  There also exists a randomized algorithm that on input an
  $n$-point set $P \subset \mathbb{Q}^d$ and supply function $\mu:P
  \to \mathbb{Z}$ such that $U = \sum_{p \in P_+}{\mu(p)}$, runs in
   time $O(n\eps^{-O(d)} \log(U)^{O(d)}\log(n)^2)$, and with
  probability at least $1/2$ finds a transportation map with cost at
  most $(1+\eps)~\cdot~\cost(P,\mu)$.
\end{theorem}
In constant dimension, the dependence of the running time on the
aspect ratio $\Delta$ or total supply $U$ is polylogarithmic, and the dependence on
the approximation $\eps$ is polynomial. The dependence on $n$ is just
$O(n\log n)$ (respectively $O(n\log(n)^2)$). This is in contrast with
prior work which either had a much larger dependence on $n$, or a
polynomial dependence on $U$. 

In the proof of this result, we employ a combination of geometric
tools, and tools from continuous optimization developed for the
general minimum cost flow problem. As a first step, we reduce the
transportation problem to an (uncapacitated) minimum cost flow problem
on a random sparse graph. This construction is closely related to the
prior work of Sharathkumar and Agarwal in the case of unit
capacities~\cite{SharathkumarA12-STOC}, and also to the estimation
algorithm for the transportation cost in~\cite{AndoniNOY14}. In
particular, the sparse graph and minimum cost flow instance we
construct are a simplification of the minimum cost flow instance
in~\cite{AndoniNOY14}. Together with the recent work of
Sherman~\cite{Sherman17}, this reduction is enough to get a
$O(n^{1+o(1)}\eps^{-O(d)})$ time algorithm to estimate the transportation
cost. As mentioned above, this running time can also be achieved using
a Euclidean spanner, and the random sparse graph we use can, in fact,
be seen as a randomized spanner. Our graph, however, has a nice
hierarchical structure not shared by other spanner
constructions. In particular, there is a quadtree of the input point
set such that any edge leaving a cell of the quadtree goes to either a
sibling cell or a parent cell. This property is useful both for
improving the running time further, and for computing a transportation
map.

To further improve the running time, we open up Sherman's elegant
framework, and combine it with classical geometric
constructions. Sherman's framework is based on repeatedly finding a
flow which is approximately optimal, and approximately feasible. This
can be achieved quickly, for example using the multiplicative weights
update method~\cite{AroraHK12}. Once an approximately feasible flow is
found, the supplies and demands are updated to equal the
still-unrouted part of the original supplies and demands, and the
problem is solved again. Finally, all the flows computed so far are
added up; if the resulting flow is very close to being feasible, it
can be ``rounded off'' to a truly feasible one with a small loss in
the cost. Sherman showed that, if the problem being solved is
sufficiently well conditioned in the appropriate way, then this
process of repeatedly finding approximately feasible flows converges
very fast to a feasible one.

Unfortunately, most minimum cost flow problems are not
well-conditioned in their natural formulations. For this reason, we
need to find a \emph{preconditioner}, i.e.~a linear transformation, to
be applied to the constraints of the minimum cost flow problem, which
makes the problem well-conditioned. This is where we can further
exploit the structure of the geometric transportation
problem. Constructing a preconditioner in our random sparse graph is
closely related to estimating the transportation cost by an embedding
into $\ell_1$.
 We exploit ideas from the known embeddings of
transportation cost into $\ell_1$~\cite{Charikar02,IndykThaper03} to
construct a preconditioner with condition number that depends
polynomially on the approximation factor and the logarithm of the
aspect ratio. The insight that these simple and well-known techniques
can be repurposed to give high-quality preconditioners is one of our
main conceptual contributions.

The final algorithm does not require any specialized data structures
and is relatively simple: once the sparse graph is constructed, and
the preconditioner is applied, we iteratively run the
multiplicative weights method several times, and, once we are close
enough to a feasible flow, greedily round off the errors. This is
sufficient for estimating the transportation cost, but does not yet
give us a transportation map, because we transformed the original
problem into a different minimum
cost flow problem on a sparse graph. In the final section of our paper
we show a simple way to efficiently transform a flow on this sparse
graph into a transportation map, without increasing its cost.

\section{Notation and Basic Notions}

We use the notation $\|x\|_p$ for the $\ell_p$ norm of a vector $x \in
\R^d$: $\|x\|_p = \left(\sum_{i = 1}^d{|x_i|^p}\right)^{1/p}$. We will
assume that the input $P \subset \mathbb{Q}^d$ to our problem lies in
$[0,\Delta]^d$, and that the smallest Euclidean distance between any
two points is at least $1$. Since any shifts, rotations, or dilations
of the points do not change the problem, this assumption is equivalent
to assuming that the aspect ratio of $P$, i.e.~the ratio between the
diameter of $P$ and the smallest distance between two distinct points,
is bounded between $\Delta$ and $\sqrt{d}\Delta$. 

In fact, at the cost of a small increase in the approximation factor,
we can reduce to this case of bounded aspect ratio. The reduction
produces a point set $P$ for which $\Delta$ is bounded by a factor
that depends on $U = \sum_{p \in P_+}{\mu(p)}$. 

\begin{lemma}[\cite{Indyk07,AndoniNOY14}]\label{lm:reduce-aspect}
  Suppose that there exists a function $T$, increasing in all its
  parameters, and an algorithm which, for any $\eps, \delta < 1$, on
  input of size $n$ in $[0, \Delta]^d \cap \Z^d$, runs in time $O(n
  T(\Delta, \eps, \delta))$, and with probability $1-\delta$ computes
  a $1+\eps$ approximation to the geometric transportation
  problem. Then there exists an algorithm that takes any input $P
  \subset \R^d$ of size $n$ and $\mu:P \to \Z$ such that $U = \sum_{p
    \in P_+}{\mu(p)}$, runs in time $O(n T(cdU/\eps^2, \eps, \delta
  n))$ for an absolute constant $c$, and, with probability $1 -
  \delta$, achieves an approximation of $1 + O(\eps)$ for the geometric
  transportation problem on $P$ and $\mu$.
\end{lemma}

\snote{Maybe sketch this for the full version, but not for the submission.}

With this lemma, the second algorithm in Theorem~\ref{thm:main}
follows from the first one. For this reason, we can focus our
presentation on the case of bounded aspect ratio.

For a set $V$, we call $f \in \R^{V \times V}$ a \emph{flow} if it
is anti-symmetric, i.e. $f_{uv} = -f_{vu}$ for every
$u,v \in V$. Intuitively, we think of $f_{uv} > 0$ as flow
going in the direction from $u$ to $v$. For a graph $G = (V, E)$, we
define a flow on $G$ to be a flow $f$ supported on $E$, i.e.~one such
that $f_{uv} = 0$ for any $(u,v) \not \in E$. The \emph{divergence} of
a flow $f$ at $u$ is the quantity $\sum_{v \in V}{f_{uv}}$, i.e.~the
excess of the flow leaving $u$ over the flow entering $u$. When the
divergence at $u$ is $0$, we say that \emph{flow conservation} is
satisfied at $u$. 

In general, we assume that our graphs are stored in the adjacency list
representation, and that flow values are stored with the adjacency
list.

\section{Reduction to Minimum Cost Flow on a Sparse Graph}
\label{sect:reduction}

The first step in our algorithm is to reduce the geometric
transportation problem, which is naturally modeled as a minimum cost
flow problem on a complete bipartite graph, to another minimum cost
flow on a sparse random graph. The following construction is a
simplified version of one used in~\cite{AndoniNOY14}.

\subsection{Graph Construction}

Recall that $P \subset [0,\Delta^d]$.  We start by constructing a grid
containing all the points in $P$ as follows. We sample a uniformly
random point $x \in [0,\Delta]^d$, and define the cell $\cell_0 =
[-\Delta, \Delta]^d + x$.  Note that all points in $P$ are in
$\cell_0$.  We say that $\cell_0$ is on {\it level 0}.  The set of
cells on level $\ell$ is labelled $\grid_\ell$ and is constructed by
taking each cell $\cell$ in $\grid_{\ell-1}$ and dividing it into
$2^d$ equally-sized cubes of side length half the side length of
$\cell$, and retaining those that contain points in $P$.  We say that
$\cell$ is the {\it parent} of the cells into which it was divided.
Conversely, those $2^d$ cells are called the {\it children} of $\cell$.
This division is continued until level $L$, where all points in $P$
lie in different cells in $\grid_L$.  With probability $1$, no point
of $P$ lands on the boundary of any cell, so we will ignore this
event.

Since the side length of any cell on level $\ell$ is half of the side
length of a cell on level $\ell-1$, any cell in $\grid_\ell$ has side
length $2^{1-\ell}\Delta$, and we refer to this length as $\Delta_l$.
Also note that, since the closest two points in $P$ are at least distance  $1$
apart, we have $L\leq\log_2(2\sqrt{d}\Delta)$, so $L = O(\log\Delta)$.
Moreover, any point $p \in P$ lies in at most $L+1$ cells, one per level,
and, since in $\grid_\ell$ we only retain cells that contain points in
$P$, we have that $|\grid_0 \cup \ldots \cup \grid_L| \le n(L+1) =
O(n\log\Delta)$. 

Let $\eps_0$ be a small number such that $\eps_0^{-1}$ is an {\it
  even} integer. This $\eps_0$ is not the same as the value used when
we speak of a $(1+\eps)$-approximation, although the two are related,
as we will see.  Next, for each $\ell \in \{0, \ldots, L$,  we take
each cell $\cell \in \grid_\ell$, and we divide it into $\eps_0^{-d}$
{\it subcells}, each of side length $\eps_0\Delta_l$. The set of all
such subcells of all $\cell \in \grid_\ell$ is
denoted $\subgrid_\ell$.  To each subcell we associate a {\it net
  point} at its centre and we denote the set of all net points on
level $\ell$ as $N_\ell$. See Figure~\ref{fig:grid} for an illustration. Note that
$|N_\ell|={|\grid_\ell|}{\eps_0^{-d}}$, so $|N_0 \cup \ldots \cup
N_L| \le {n(L+1)}{\eps_0^{-d}} = O({n}{\eps_0^{-d}}\log\Delta)$.

\begin{figure}[tp]
  \centering
  \includegraphics{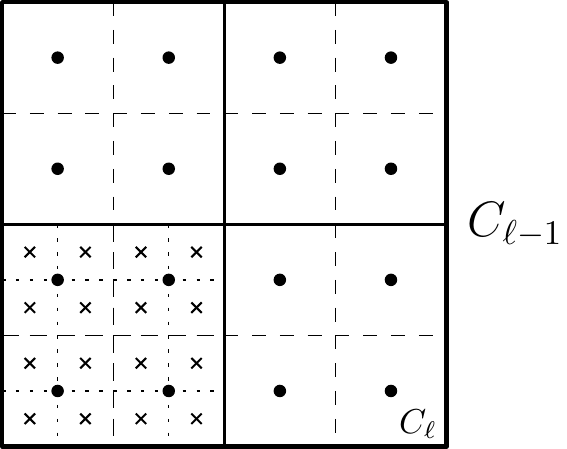}
  \caption{A cell $C_{\ell-1}$ on level $\ell-1$ and the four cells on level $\ell$ contained in it. The subcells of $C_{\ell-1}$ are shown with dashed lines, and the corresponding net points in $N_{\ell-1}$ with black dots. The subcells of one of the cells on level $\ell$ are shown with dotted lines, and the corresponding net point in $N_\ell$ with crosses. }
  \label{fig:grid}
\end{figure}

If $p$ is a point in $\cell_0$, let $\cell_\ell(p)\in\subgrid_\ell$ be the unique subcell on
level $\ell$ that contains $p$.
Similarly, let $N_\ell(p)$ be the net point on level $\ell$ closest to
$p$. Equivalently, $N_\ell(p)$ is the center of $\cell_\ell(p)$.
We say that $N_\ell(p)$ and $N_{\ell-1}(p)$ are each other's {\it child}
and {\it parent}, respectively.

We are now ready to construct our graph $G=(V,E)$.  We set $V=P\cup
N_0\cup N_1\cup\dots\cup N_L$, the set of all of the points in the
original problem and all of the netpoints we have constructed.  By the
previous discussion, $|V| = O({n}{\eps_0^{-d}}\log\Delta)$.  We build
our graph using three types of edges.  The first set, $E_1$, connects
points $p$ in $P$ to their closest net point on level $L$,
i.e.~$E_1=\{(p,N_L(p)):  p\in P\}$.  The size of this set is
$|E_1| = n$. The second set of edges, $E_2$, connects all net points
of subcells of a given cell pairwise.  Thus, 
\[
E_2=\{(u,v):\ell\in\{0,1,\dots,L\},\cell\in\grid_\ell,
u,v\in N_\ell\cap\cell~\text{s.t.}~u\neq v\}.
\] 
Since $|\grid_\ell|
\le n$, and, for any $\cell\in\grid_\ell$, $|N_\ell \cap \cell| =
\eps_0^{-d}$, we have that $|E_2| \le (L+1)n\eps_0^{-2d} =
O(n\eps_0^{-2d}\log\Delta)$.  The last set of edges, $E_3$, connects
net points to their parent net points, i.e.
\[
E_3=\{(N_{\ell-1}(u),u): \ell\in\{1,2,\dots,L\}, u \in N_{\ell}\}.
\] 
The size of $E_3$ is less than the total number of net
points, so $|E_3| = O({n}{\eps_0^{-d}}\log\Delta)$.  Then, the set of
edges is defined as $E=E_1\cup E_2\cup E_3$, and this completes the
description of $G = (V,E)$. The number of edges is
dominated by the size of $E_2$, and is in
$O(n\eps_0^{-2d}\log\Delta)$.  We will be calculating distances
between points along paths in this graph, and we define the
distance/cost of any edge to equal the Euclidean distance between the
two endpoints of the edge.

For the rest of the paper we will use $G$ to refer to the random graph
constructed above. Having stored the random shift of $\cell_0$, we can
enumerate the vertices in time $O({n}{\eps_0^{-d}}\log\Delta)$ by
going through each point in $P$ and checking which cells it lies
in. We can then store the vertices in a static
dictionary~\cite{FKS84}. Similarly, we can construct the adjacency
lists of the vertices in time $O(n\eps_0^{-2d}\log\Delta)$. 

\subsection{Embedding Euclidean Distance into a Sparse Graph}
\label{sect:embedding}

\newcommand{\dist}{\mathrm{dist}}

Let $\dist_G(u,v)$ be the shortest path distance between two nodes $u$
and $v$ of $G$, i.e.
\[
\dist_G(u,v) = \min \sum_{i = 1}^{k} \|u_i -u_{i-1}\|_2,
\]
 with the minimum taken over paths $u = u_0, \ldots, u_k
= v$ connecting $u$ and $v$ in $G$.  Our goal for this section is to
show that, for any two $p$ and $q$ in $P$, the expected value of
$\dist_G(u,v)$ is close to the Euclidean distance between
them. In other words, we want to show that our random graph
construction provides a randomized embedding of Euclidean distance
into the shortest path metric of a sparse graph. This is
similar to the embeddings used in \cite{AndoniNOY14} and
\cite{SharathkumarA12-SODA,SharathkumarA12-STOC}, and can be seen as a
randomized spanner construction, in which the graph $G$ only
needs to have low stretch in expectation. 

For two points $p$ and $q$ in $P$, we define $\ell(p,q)$ to be the
level of the smallest cell containing both $p$ and $q$, which can
depend on the random shift of the grid.  This means that
$\cell_{\ell(p,q)}(p)=\cell_{\ell(p,q)}(q)$, but
$\cell_{\ell(p,q)+1}(p)\neq\cell_{\ell(p,q)+1}(q)$.  The following
definition will be useful.

\begin{definition}
  Given two points $p,q\in P$, the {canonical path} between $p$ and
  $q$ in the random graph $G$ consists of the edges 
  \begin{multline*}
  \{(N_\ell(p),  N_{\ell-1}(p):  \ell(p,q) + 1 \le \ell \le L\}
  \cup
  \{(N_{\ell(p,q)}(p), N_{\ell(p,q)}(q))\}
  \\
  \cup
  \{(N_\ell(q),  N_{\ell-1}(q): \ell(p,q) + 1 \le \ell \le L\}
  \cup
  \{(N_L(p), p), (N_L(q),q)\}.
  \end{multline*}
\end{definition}

The canonical path could be much longer than $\|p-q\|_2$ because, if
$p$ and $q$ are two nearby points that are separated by a grid line,
the canonical path will be much larger than the Euclidean distance.
However, when $p$ and $q$ are very close together, the likelihood
that such a grid line will fall between them is low.  To
formalize this intuition, we need the following well-known bound on
the probability that $p$ and $q$ will be split at a certain level of
the tree, which is elementary, and goes back at least to Arora's work
on Euclidean TSP~\cite{Arora98}.

\begin{lemma}\label{lm:sep-prob}
Over the random shift of $\cell_0$, $\forall p,q\in P$,
\[\Pr(\cell_\ell(p)\neq\cell_\ell(q))\le
\frac{\|p-q\|_1}{\Delta_\ell} 
\le \frac{\sqrt{d}\|p-q\|_2}{\Delta_\ell}.\]
\end{lemma}


Using this standard lemma, we can prove the bound on the distortion
between Euclidean distance and the expected shortest path distance in
$G$. The proof, which is standard, is deferred to the appendix.

\begin{lemma}\label{lm:distort}
  Let $p,q\in P$. Then, for any shift of $\cell_0$, $\|p-q\|_2 \le
  \dist_G(p,q)$. Moreover, over the random shift of $\cell_0$,  
  $\E[\dist_G(p,q)] \le (1 + O(\eps_0L)) \|p-q\|_2$. 
\end{lemma}

\subsection{Approximating the Transportation Cost with Minimum Cost
  Flow}
\label{sect:min-cost-flow}

We are now ready to translate the transportation problem to a minimum
cost flow problem in $G$, captured by the following linear program. 
\begin{align}
  \text{Minimize } &\sum_{(u,v) \in E} f_{uv} \|u - v\|_2\label{eq:mincost-obj}\\
  \text{subject to}\notag\\
  \forall u,v \in V: &f_{uv} = -f_{vu}\\
  \forall p \in P: &\sum_{u \in V}{f_{pu}} = \mu(p),\label{eq:mincost-supply}\\
  \forall u \in V \setminus P: &\sum_{v \in v}{f_{uv}} = 0.\label{eq:mincost-fc}
\end{align}
Let us denote by $\cost(G, \mu)$ the value of the optimal solution to
the problem~\eqref{eq:mincost-obj}--\eqref{eq:mincost-fc}. Note that,
since $G$ is a random graph, $\cost(G,\mu)$ is also a random
variable. 

The next theorem shows that that $\cost(G,\mu)$ approximates the transportation
cost $\cost(P,\mu)$. 

\begin{theorem}\label{thm:mincost-apx}
  For any shift of $\cell_0$, $\cost(P,\mu)  \le \cost(G,\mu)$. Moreover,
  over the random shift of $\cell_0$,
  \[
  \E[\cost(G,\mu)] \le (1 + O(\eps_0L)) \cost(P,\mu).
  \]
\end{theorem}
\begin{proof}[Proof of Theorem~\ref{thm:mincost-apx}]
  To prove the first claim, we show that a flow ${f}$ satisfying
  \eqref{eq:mincost-supply}--\eqref{eq:mincost-fc} can be transformed
  into a transportation map ${f}'$ for $P$ and $\mu$ without
  increasing its cost. Towards this goal, let us decompose $f$ into a
  sum of paths as $f = \sum_{i = 1}^k{a_i f_{\pi_i}}$, where $\pi_i$
  is a path from some point $p_i$ in $P_+$ to some point $q_i$ in $P_-$,
  $f_{\pi_i}$ is the unit flow traveling on edges of $\pi_i$ in the
  direction from $p_i$ to $q_i$, and $a_1, \ldots, a_k$ are non-negative
  integers. The existence of such a decomposition is a standard fact,
  see e.g.~Chapter 10 of \cite{Schrijver-volA}. The cost of
  $f$ is at least $\sum_{i = 1}^k{a_i \dist_G(p_i, q_i)}$. We can
  construct $f'$ by defining $f'_{pq} = \sum_{i: p_i = p, q_i =
    q}{a_i}$. The constraints on $f$ imply that $f'$ is a
  feasible transportation map; its cost is 
  \[
  \sum_{i = 1}^k{a_i \|p_i - q_i\|_2} \le \sum_{i = 1}^k{a_i \dist_G(p_i, q_i)},
  \]
  where the inequality follows from Lemma~\ref{lm:distort}, and the
  right hand side is, as already observed, at most the cost of
  $f$. The first claim then follows since we can take $f$ to be an
  optimal solution to \eqref{eq:mincost-obj}--\eqref{eq:mincost-fc}. 

  For the second claim, we need to go in reverse. Let $f'$ be an
  optimal transportation map,  i.e.~one achieving $\cost(P,\mu)$. Let
  $(p_1, q_1), \ldots, (p_k, q_k)$ be the edges in the support of
  $f'$, where, for every $i$, $p_i \in P_+$ and $q_i \in P_-$. Let
  $\pi_i$ be a shortest path in $G$ between $p_i$ and $q_i$. We
  define a flow $f$ in $G$ by $f = \sum_{i = 1}^k{f'_{p_i,q_i} f_{\pi_i}}$, where
   $f_{\pi_i}$ is, as before, a unit flow from $p_i$ to
  $q_i$ supported on the edges of $\pi_i$. The feasibility of $f$
  follows from the fact that $f'$ is a feasible transportation
  map. Moreover, the cost of $f$ is $\sum_{i =
    1}^k{f'_{p_i,q_i}\dist_G(p_i,q_i)}$. We then have
  \begin{align*}
    \E[\cost(G,\mu)] &\le \sum_{i = 1}^k{f'_{p_i,q_i}\E[\dist_G(p_i,q_i)]}\\
    &\le (1 + O(\eps_0L)) \sum_{i = 1}^k{f'_{p_i,q_i}\|p_i-q_i\|_2}\\
    &= (1+O(\eps_0L)) \cost(P,\mu),
  \end{align*}
  where the final inequality follows from Lemma~\ref{lm:distort}.
\end{proof}

\section{Solving the Minimum Cost Flow Problem}
\label{sect:preconditioner}

In this section we describe the generalized preconditioning framework
of Sherman~\cite{Sherman17}, and describe how to use it to solve the
minimum cost flow problem \eqref{eq:mincost-obj}--\eqref{eq:mincost-fc}
in nearly linear time. 

\subsection{Generalized Preconditioning Framework}

We take a short detour to describe the generalized preconditioning
framework of Sherman. Suppose that $\|\cdot\|$ is a norm on $\R^m$,
and we are given inputs $A \in \R^{n \times m}$ and $b \in
\R^n$. The framework is concerned with solving the following
optimization problem over the variables $f \in \R^m$:
\begin{align}
  \text{Minimize } &\|f\|\label{eq:minnorm-obj}\\
  \text{subject to \ \ \ }
  &Af = b,\label{eq:minnorm-constr}.
\end{align}
This formulation can capture many important problems, including,
crucially for us, the uncapacitated minimum cost flow
problem,\footnote{Uncapacitated minimum cost flow is commonly known as
  the transportation problem. We do not use this terminology, in order
  to avoid confusion with the geometric transportation problem.} as we
explain later.

A central definition in the framework is that of the nonlinear
condition number. Before we state it, let us recall that the norm of a
linear map $T:X \to Y$, where $(X, \|\cdot\|_X)$ and $(Y,
\|\cdot\|_Y)$ are two finite dimensional normed vector spaces, is
defined as
\[
\|T\| = \sup_{x \neq 0} \frac{\|T(x)\|_Y}{\|x\|_X}.
\]
We adopt the same definition for non-linear maps, as well.
For an $n\times m$ matrix $A$, and norms $\|\cdot\|_X$ on $X = \R^n$
and $\|\cdot\|_Y$ on the column span $Y$ of $A$, we use the notation
$\|A\|_{\|\cdot\|_X \to \|\cdot\|_Y}$ to refer to the norm of the
linear map from $X$ to $Y$ represented by $A$. 

Having recalled these notions, we proceed to define the condition number.
\begin{definition}\label{def:condnum}
  Let $(X, \|\cdot\|_X)$ and $(Y, \|\cdot\|_Y)$ be two
  finite dimensional normed vector spaces.
  The non-linear condition number $\kappa$ of a linear map $T: X \to Y$ is
  defined by
  \[\kappa(T) = \inf_S \|T\| \|S\|,\] 
  where $S:Y \to X$ ranges over all (not necessarily linear) maps such
  that for all $x \in X$ we have $T(S(Tx)) = Tx$.

  For an $n \times m$ matrix $A$, and norms $\|\cdot\|_X$ on $X = \R^m$
  and $\|\cdot\|_Y$ on the column span $Y$ of $A$, we define
  $\kappa_{\|\cdot\|_X \to \|\cdot\|_Y}(A)$ as $\kappa(T)$, where
  $T: X \to Y$ is the linear map represented by the matrix $A$.
\end{definition}
This definition generalizes the standard condition number, which is
the special case in which both norms are taken to be Euclidean.

The generalized preconditioning framework is based on composing rough approximation
algorithms to get a high-quality approximation. The rough
approximation algorithms are further allowed to
violate the constraints of the problem slightly. Thus, they achieve a bi-criteria
approximation, captured by the following definition.
\begin{definition}\label{def:bicriteria}
  Let $\|\cdot\|_Y$ be a norm defined on the range of the matrix $A$
  in \eqref{eq:minnorm-constr}. Let $f^\ast$ be an optimal solution to
  the problem \eqref{eq:minnorm-obj}--\eqref{eq:minnorm-constr}. Then
  $f \in \R^m$ is called an $(\alpha, \beta)$-solution 
  (with respect to $\|\cdot\|_Y$)
  to
  \eqref{eq:minnorm-obj}--\eqref{eq:minnorm-constr} if $\|f\| \le
  \alpha \|f^\ast\|$ and, moreover, 
  \[
  \|Af - b\|_Y\le \beta {\|A\|_{\|\cdot\| \to \|\cdot\|_Y} \|f^\ast\|}.
  \]
  An algorithm that, when given inputs $A$ and $b$, outputs an
  $(\alpha, \beta)$-solution $f$ is called an $(\alpha,
  \beta)$-solver (with respect to $\|\cdot\|_Y$) for
  \eqref{eq:minnorm-obj}--\eqref{eq:minnorm-constr}. 
\end{definition}

We use the next result of Sherman, which gives a solver for
\eqref{eq:minnorm-obj}--\eqref{eq:minnorm-constr} with running time
controlled by condition number of the constraint matrix. 
\begin{theorem}[\cite{Sherman17}]\label{thm:1to1-opt}
  Let $\eps, \beta > 0$, and suppose that the norm in
  \eqref{eq:minnorm-obj} is the $\ell_1^m$ norm. Let $\kappa =
  \kappa_{\|\cdot\|_1 \to \|\cdot\|_1}(A)$, and let $M$ be an upper
  bound on the time necessary to compute matrix-vector products with
  $A$ and $A^\top$.  Then there exists a $(1+\eps, \beta)$-solver with
  respect to $\|\cdot\|_1$ for the problem
  \eqref{eq:minnorm-obj}--\eqref{eq:minnorm-constr} with running time
  bounded by
  \[
  O\left(\kappa^2(m+ n + M)\log(m)(\eps^{-2} +  \log(1/\beta))\right).
  \]
\end{theorem}
This algorithm is based on repeatedly applying an $(\alpha, \beta)$
solver with much worse dependence on $\beta$; in fact, at the cost of
a slightly worse dependence on $\kappa$, for the latter one can use a
simple solver based on the multiplicative weights update method. To
make our paper more self-contained, we sketch this latter algorithm in
the appendix.

Rescaling $f$ coordinatewise gives us the following easy
corollary.
\begin{corollary}\label{cor:cto1-opt}
  Let $\eps, \beta > 0$, and let $c \in \R_{> 0}^m$. Suppose that
  the norm in \eqref{eq:minnorm-obj} is given by $\|f\|_c = \sum_{i =
    1}^m{c_i f_i}$. Let $\kappa = \kappa_{\|\cdot\|_c \to
    \|\cdot\|_1}(A)$, and let $M$ be the time necessary to compute
  matrix-vector products with $A$ and $A^\top$. 
  Then there exists a $(1+\eps, \beta)$-solver with respect to
  $\|\cdot\|_1$ for the problem
  \eqref{eq:minnorm-obj}--\eqref{eq:minnorm-constr} with running time
  bounded by
  $O\left(\kappa^2(m+ n + M)\log(m)(\eps^{-2} +  \log(1/\beta))\right).$
\end{corollary}
\begin{proof}
  Define the diagonal matrix $C$ by $c_{ii} = c_i$. 
  The problem   \eqref{eq:minnorm-obj}--\eqref{eq:minnorm-constr} with
  norm $\|\cdot\| = \|\cdot\|_c$ can be reformulated, after the change
  of variables $g = Cf$, as
  \begin{align*}
    \text{Minimize } &\|g\|_1\\
    \text{subject to \ \ \ }
    &AC^{-1}g = b,
  \end{align*}
  Then the corollary follows from Theorem~\ref{thm:1to1-opt} since
  $\kappa_{\|\cdot\|_c \to \|\cdot\|_1}(A) = \kappa_{\|\cdot\|_1 \to
    \|\cdot\|_1}(AC^{-1})$. The latter follows from two observations:
  first, $\|AC^{-1}\|_{\|\cdot\|_1 \to \|\cdot\|_1} =\|A\|_{\|\cdot\|_c \to \|\cdot\|_1}$; 
  second, the map $G$ from
  the column span $Y$ of $A$ to $\R^m$ satisfies $AGAf = Af$ for all
  $f$ if and only if $AC^{-1}(CG(AC^{-1}g))$ for all $g$, and 
  $\|CG\|_{\|\cdot\|_1 \to \|\cdot\|_1} =\|G\|_{\|\cdot\|_1 \to
    \|\cdot\|_c}$. 
\end{proof}

While Theorem~\ref{thm:1to1-opt} and Corollary~\ref{cor:cto1-opt} allow
us to find solutions which are very close to being feasible, they do
not give an exactly feasible solution. The final result we use from
Sherman allows us to use a solver which exactly satisfies the
constraints, but only achieves a large approximation ratio, to round
an approximately feasible solution to an exactly feasible one. 

To state the result, we need to define the composition of two solvers
$\mathcal{F}$ and $\mathcal{F}'$. The composed algorithm $\mathcal{F}'
\circ \mathcal{F}$ first calls $\mathcal{F}$ on $A$
and $b$ to get $f$; then it calls $\mathcal{F}'$ on $A$ and
the residual vector $b - Af$ to get a solution $f'$; finally it
outputs $f + f'$. 

\begin{theorem}[\cite{Sherman17}]\label{thm:composition}
  Let $C \ge 1$, and $\eps, \beta > 0$.
  Let $\kappa = \kappa_{\|\cdot\| \to \|\cdot\|_Y}(A)$, where
  $\|\cdot\|$ is the norm in \eqref{eq:minnorm-obj}, and
  $\|\cdot\|_Y$ is some norm on the column span of $A$. Then, if
  $\mathcal{F}$ is a $(1+\eps, \eps\beta/\kappa)$-solver, and
  $\mathcal{F}'$ is a $(C, 0)$-solver, the composition $\mathcal{F}'
  \circ \mathcal{F}$ is a $(1 + \eps + C\eps\beta, 0)$-solver. 
\end{theorem}

\subsection{Preconditioner Construction}

The formulation \eqref{eq:minnorm-obj}--\eqref{eq:minnorm-constr}
captures the (uncapacitated) minimum cost flow problem with arbitrary
demands. We will explain this for the graph $G = (V,E)$ defined in
Section~\ref{sect:embedding} and the minimum cost flow problem
\eqref{eq:mincost-obj}--\eqref{eq:mincost-fc}. Let us pick an arbitrary
orientation on the edges $E$, and call the directed edges
$\vec{E}$. Then take $A$ to be the directed vertex by edge incidence
matrix of $G$: it is indexed by $V \times \vec{E}$, and for any node
$u$, and any directed edge $e = (v,w)$, set $A_{u,e} = 1$ if $u = v$,
$A_{u,e} = -1$ if $u = w$, and $A_{u,e} = 0$ otherwise. We represent a
flow $f$ by its restriction to $\vec{E}$, seen as a vector in
$\R^{\vec{E}}$. Slightly abusing notation, we will use the letter $f$
both for this vector, and for the flow, which is defined both for
$(u,v) \in \vec{E}$, and for $(v,u)$, with $f_{vu} = -f_{uv}$.  For a
flow vector $f$, the product $Af$ gives us the vector of divergences,
i.e.~for any $u \in V$
\[
(Af)_u = \sum_{v: (u,v) \in \vec{E}}{f_{uv}} - \sum_{v: (v,u) \in \vec{E}}{f_{vu}}
= \sum_{v: (u,v) \in E}f_{uv}.
\]
We define the vector $b \in \R^V$ to encode the supplies, i.e.~for $p
\in P$ we set $b_p = \mu(p)$, and for $u \in V \setminus P$ we set
$b_u = 0$.  It follows that the constraint $Af = b$ encodes
\eqref{eq:mincost-supply}--\eqref{eq:mincost-fc}. Finally, let us
denote the cost of an edge $e = (u,v) \in \vec{E}$ by $c(e) = \|u
- v\|_2$. Then, for the norm in the objective \eqref{eq:minnorm-obj},
we choose $\|f\|_c = \sum_{e \in E}{c(e) |f_e|}$. With these choices
of $A$, $b$, and $\|\cdot \| = \|\cdot\|_c$, an optimal solution to
\eqref{eq:minnorm-obj}--\eqref{eq:minnorm-constr} gives an optimal
solution to \eqref{eq:mincost-obj}--\eqref{eq:mincost-fc}.
For the rest of this section we will fix $A$, $b$, $c$, and
$\|\cdot\|_c$ to be as just defined. 

Unfortunately, we cannot directly use Corollary~\ref{cor:cto1-opt} to
get the running time we are aiming for, since the condition number of
the matrix $A$ could be large. We address this by designing a
preconditioner: another matrix $B$, of full column rank, which can be
applied quickly, and has the property that $BA$ has small condition
number. This allows us to apply Corollary~\ref{cor:cto1-opt} to the
problem \eqref{eq:minnorm-obj}--\eqref{eq:minnorm-constr} with the
modified, but equivalent, constraint $BAf = Bb$, and get a fast
algorithm for the minimum cost flow problem in $G$.

In our construction of the preconditioner $B$ we will use the
following lemma, which was implicit in \cite{Sherman17}. 
\begin{lemma}\label{lm:apx-to-kappa}
  Let $\|\cdot\|_X$ be a norm on $\R^n$, and let $H$ be an $n \times
  m$ matrix. Suppose that there exists a $\gamma > 0$ such that for
  any $h$ in the column span $Y$ of $H$, the norm
  $\|\cdot\|_Y$ satisfies
  \[
  \|h\|_Y \le \min\{\|f\|_X: f \in \R^n, Hf = h\} \le \gamma \|h\|_Y.
  \]
  Then $\kappa_{\|\cdot\|_X \to \|\cdot\|_Y}(H) \le \gamma$. 
\end{lemma}
\begin{proof}
  Let us identify $H$ with the linear map from $X = \R^n$ to $Y$
  represented by it. Define $S(h) = \arg \min\{\|f\|_X: Hf = h\}$, with ties
  between minimizers broken arbitrarily. It is clear that for any $f$,
  $HS(Hf) = Hf$. By assumption, we have the inequalities
  \begin{align*}
    \|Hf\|_Y &\le \|f\|_X,\\
    \|S(h)\|_X &\le \gamma \|h\|_Y,
  \end{align*}
  which are equivalent to $\|H\| \le 1$ and $\|S\| \le \gamma$. 
\end{proof}

We will design a preconditioner matrix $B$ such that $\|\tilde{b}\|_1$
approximates the cost of the minimum cost flow with supply vector
$\tilde{b}$. Then the fact that $BA$ has small condition number will follow
from Lemma~\ref{lm:apx-to-kappa}. The following lemma captures the
construction of $B$, which is inspired by embedding of the Earth
Mover Distance in $\ell_1$~\cite{Charikar02,IndykThaper03}.

\begin{lemma}\label{lm:preconditioner}
  There exists a matrix $B \in \R^{V \times V}$ of full column rank with at most $O(|V|\log\Delta)$
  nonzero entries, such that the following holds.  For any $\tilde{b}
  \in \R^V$ such that $\sum_{u \in V}{\tilde{b}_u} = 0$, we have
  \[
  \|B\tilde{b}\|_1 \le \min\{\|f\|_c: f \in \R^{\vec{E}}, Af = \tilde{b}\} \le \gamma \|B\tilde{b}\|_1,
  \]
  for $\gamma = O(\log(\Delta)/\eps_0)$. Moreover, a flow $f$
  satisfying $Af = \tilde{b}$ of cost at most $\gamma\|B\tilde{b}\|_1$
  can be constructed in time $O({n}{\eps_0^{-d}}\log\Delta)$ given
  $B\tilde{b}$.
\end{lemma}
\begin{proof}
  The matrix $B$ has one row associated with each vertex $u \in
  V$.  For $p \in P$, we let $B_{p,p} = \|p - N_L(p)\|_2$
  and $B_{p,u} = 0$ for any $u \neq p$. 
  Every other vertex $u \in V
  \setminus P$ is a net point. 
  Suppose that $u$ is in $N_\ell$; then for any $v \in V$
  such that $N_\ell(v) = u$, $B_{u,v} = \frac{\eps_0
    \Delta_\ell}{4(L+1)}$, and for all other $v$, $B_{u,v} = 0$. This
  defines the matrix $B$. Notice that each vertex $u \in V$
  contributes to at most $L+2$ nonzero entries of $B$: one for each
  $\ell \in \{0, \ldots, L\}$ corresponding to
  $N_\ell(u)$, and one more when $u\in P$. It is easy to
  verify directly that $B$ has full column rank. 

  The definition of $B$ guarantees that
  \begin{equation}\label{eq:emd-sketch}
  \|B\tilde{b}\|_1
  = 
  \sum_{p \in P}{|\tilde{b}_p|\|p - N_L(p)\|_2} +
  \frac{1}{L+1}\sum_{\ell = 0}^L  \frac{\eps_0 \Delta_\ell}{4}
  \sum_{\subcell \in \subgrid_\ell}
  \left|\sum_{u \in \subcell\cap V}\tilde{b}_u\right|.
  \end{equation}
  The remainder of the proof is broken down into several claims. For
  the first claim, let us extend the definition of
    $\cost(G, \cdot)$ to supplies which can be non-zero on net points
    as well by \[\cost(G, \tilde{b}) = \min\{\|f\|_c: f \in \R^{\vec{E}}, Af =
    \tilde{b}\}.\] 
  \begin{claim}
    $\|B\tilde{b}\|_1 \le \cost(G, \tilde{b})$.
  \end{claim}
  \begin{subproof}
    In any feasible flow $f$, the total cost of the flow on edges
    incident to points $p \in P$ must equal the first term in~\eqref{eq:emd-sketch}. Then, we
    just need to show that the remaining terms are a lower bound on
    the total cost of the flow on the remaining edges.  In fact we
    will show that, for each $\ell$,
    \begin{equation}\label{eq:precond-lb}
      \frac{\eps_0 \Delta_\ell}{4}
      \sum_{\subcell \in \subgrid_\ell}  \left|\sum_{u \in \subcell\cap V}\tilde{b}_u\right|
      \le
      \cost(G, \tilde{b}) -   \sum_{p \in P}{|\tilde{b}_p|\|p - N_L(p)\|_2}.
    \end{equation}
    It follows that the average of the term on the left over all
    $\ell$ is also a lower bound on the right hand side, which proves
    the claim. To establish inequality \eqref{eq:precond-lb} notice
    that, in any feasible $f$, flow of value at least $\left|\sum_{u
        \in \subcell\cap V}\tilde{b}_u\right|$ must enter or leave
    each $\subcell \in \subgrid_\ell$, and every edge $e$ from a
    vertex in $\subcell$ to a vertex outside of $\subcell$ has cost at
    least $c(e) \ge \eps_0 \Delta_\ell/2$ by the definition of
    $G$. Therefore, the total cost of flow leaving or entering
    $\subcell$ is at least
    $\frac{\eps_0 \Delta_\ell}{2} \left|\sum_{u \in \subcell\cap
        V}\tilde{b}_u\right|$. Adding up these terms over all
    $\subcell \in \subgrid_\ell$ accounts for the cost of any edge of
    $G$ at most twice, and \eqref{eq:precond-lb} follows.
  \end{subproof}

  \begin{claim}
    $\cost(G, \tilde{b}) \le  \gamma\|B\tilde{b}\|_1$
    for $\gamma = O(L/\eps_0) = O(\log(\Delta)/\eps_0)$.
  \end{claim}
  \begin{subproof}
    We prove the claim by constructing an appropriate feasible flow
    $f$.  We proceed to construct $f$ by levels, from the bottom
    up. The construction will be efficient, also proving the claim
    after ``moreover'' in the statement of the lemma. On level $L$,
    for any $p \in P$, we set $f_{pu} = \tilde{b}_p$, where $u =
    N_L(p)$. Then, for levels $\ell = L-1, L-2, \ldots, 1$, we execute
    the following procedure in every subcell $\subcell \in
    \subgrid_{\ell-1}$. Let us define, for any node $u$ and the
    current flow $f$, the surplus
    \[
    \delta(u,f) = 
    \sum_{v: (u,v) \in \vec{E}}{f_{uv}} - \sum_{v: (v,u) \in \vec{E}}{f_{vu}}
    - \tilde{b}_u
    =   \sum_{v: (u,v) \in {E}}{f_{uv}}
    - \tilde{b}_u.
    \]
    I.e.~this is how much more flow leaves $u$ than should, as
    prescribed by $\tilde{b}$. 
    Pick any two $u$ and $v$ in $N_\ell\cap \subcell$ such that
    $\delta(u,f) > 0 > \delta(v,f)$, and add $\min\{|\delta(u,f)|,
    |\delta(v,f)|\}$ units of flow to $f_{vu}$. Continue until we have
    that for all $u \in N_\ell \cap \subcell$ the surpluses
    $\delta(u,f)$ have the same sign. Since at every step of this procedure we make
    the surplus of at least one node $0$, and we always decrease the
    absolute value of the surplus at any node, we will stop after
    at most $|N_\ell \cap \subcell|$ steps. 
    Finally, for the node $u\in N_{\ell-1}$ which is the center of
    $\subcell$, and for all nodes $v \in N_\ell \cap \subcell$ for which
    $\delta(v,f) \neq 0$, we set $f_{uv} = \delta(v,f)$. After this
    final step, every node $u \in N_\ell \cap \subcell$ has surplus
    $0$. 

    For $\ell = 0$, we perform essentially the same procedure, but in
    the entire cell $\cell_0$. I.e.~we pick any two $u$ and $v$ in
    $N_\ell\cap \cell_0$ such that $\delta(u,f) > 0 > \delta(v,f)$, and
    add $\min\{|\delta(u,f)|, |\delta(v,f)|\}$ units of flow to
    $f_{vu}$. Once again, after at most $|N_0|$ steps all surpluses will
    have the same sign, and, as we show below, will in fact be $0$, so
    $f$ will be feasible. 
    
    An easy induction argument shows that, once we have processed levels
    $L, L-1, \ldots, \ell$, for any $\ell > 0$, the surplus at any node
    $u \in N_L \cup \ldots \cup N_\ell$ is $0$, and the same is true for
    the surplus at any $p \in P$. To show that the flow $f$ is feasible,
    it is enough to show that after processing the cell $C_0$ on level
    $0$, all nodes have surplus $0$. Notice that, while processing any
    subcell $\subcell$, we do not change the total surplus $\sum_{u \in
      V \cap \subcell}{\delta(u,f)}$.  This means that, for any $\ell
    >0$, after having processed a subcell $\subcell \in
    \subgrid_{\ell-1}$ with center $u$, we have 
    \begin{equation}\label{eq:surplus}
      {\delta(u,f)} = -\sum_{v \in  \subcell \cap V}{\tilde{b}_v}
      = -\sum_{v: N_{\ell-1}(v) = u}{\tilde{b}_v}.
    \end{equation}
    In particular, we have that before we
    start processing $\cell_0$, $\sum_{u \in N_0}{\delta(u,f)} = -\sum_{u
      \in V}{\tilde{b}_u} = 0$.  Since, once again, every step we make
    during the processing of $\cell_0$ preserves the total surplus, and
    we stop when all surpluses have the same sign, then at the time we
    are done it must be the case that every node has surplus $0$, and,
    therefore, $f$ is a feasible flow.
  
    Next we bound the cost of the flow $f$ constructed above.
    The first term on the right hand side of \eqref{eq:emd-sketch}
    accounts exactly for the cost of the flow on edges incident on
    points $p \in P$, as we already observed. 
    Let $0 < \ell \le L$.
    During each step executed while processing the subcell
    $\subcell \in \subgrid_{\ell-1}$, we pick some $u \in \subcell \cap
    N_\ell$, and add flow of value $\delta(u,f)$ along an edge of cost
    at most the diameter of $\subcell$, which is $\sqrt{d}\eps_0
    \Delta_\ell$. Moreover, the absolute value of $\delta(u,f)$ never
    increases while processing $\subcell$, and by \eqref{eq:surplus} it is
    bounded from above by 
    $\left|\sum_{v: N_{\ell}(v) = u}{\tilde{b}_u}\right|$ before
    $\subcell$ is processed. 
    Therefore, the cost added to the flow while processing
    $\subcell$ is bounded by 
    \[
    \sqrt{d}\eps_0\Delta_\ell  
    \sum_{u \in \subcell \cap N_\ell}\left|\sum_{v: N_{\ell}(v) = u}{\tilde{b}_u}\right|. 
    \]
    Adding up these upper bounds on the cost over all subcells in
    $\subgrid_{\ell-1}$ gives us that the total cost incurred while
    processing subcells on level $\ell-1$ is at most
    \[
    Z_\ell = \sqrt{d}\eps_0\Delta_\ell  
    \sum_{\subcell \in \subgrid_\ell}  \left|\sum_{u \in \subcell\cap V}\tilde{b}_u\right|.
    \]
    An analogous argument shows that the total cost incurred while
    processing $\cell_0$ is at most 
    \[
    Z_0 = \sqrt{d}\Delta_0  
    \sum_{\subcell \in \subgrid_0}  \left|\sum_{u \in \subcell\cap V}\tilde{b}_u\right|.
    \]
    Therefore, the cost of $f$ is bounded by 
    \[
    \sum_{p \in P}{|\tilde{b}_p|\|p - N_L(p)\|_2} + Z_L + \ldots + Z_0
    \le
    \frac{4\sqrt{d}(L+1)}{\eps_0} \|B\tilde{b}\|_1,
    \]
    where the inequality follows by a term by term comparison with
    \eqref{eq:emd-sketch}. 
  \end{subproof}
  
  The two claims finish the proof of the lemma, together with the
  observation that the construction of the flow $f$ can be implemented
  recursively in time $O({n}{\eps_0^{-d}}\log\Delta)$.
\end{proof}

Our main result for solving the minimum cost flow problem in $G$
follows.
\begin{theorem}\label{thm:flow}
  A flow $f$ feasible for the problem
  \eqref{eq:mincost-obj}--\eqref{eq:mincost-fc}, with cost at most
  $1+O(\eps)$ factor larger than the optimal cost, can be computed in time
  \[
  O\left(
    \frac{\log(\Delta)^2}{\eps^2_0}(|E|+
    |V|\log(\Delta))\log(|E|)\Biggl(\frac{1}{\eps^{2}} 
    +  \log\Bigl(\frac{\log(\Delta)}{\eps_0}\Bigr)\Biggr)
  \right).
  \]
\end{theorem}
\begin{proof}
  We reformulate the problem as minimizing $\|f\|_c$ over $f \in
  \R^{\vec{E}}$ satisfying $BAf = Bb$.  Since $B$ has full column
  rank, the constraint is equivalent to the original constraint $Af =
  b$. Moreover, the time for applying $BA$ is bounded by the total
  number of nonzeros of $B$ and $A$, which is $O(|E| + |V|\log
  \Delta)$. By Lemma~\ref{lm:preconditioner} the assumption of
  Lemma~\ref{lm:apx-to-kappa} is satisfied with $H = BA$, $h = Bb$,
  $\|\cdot\|_X = \|\cdot\|_c$ and $\|\cdot\|_Y =
  \|\cdot\|_1$. Therefore, $\kappa_{\|\cdot\|_c \to \|\cdot\|_1}(H)
  \le \gamma = O(\log(\Delta)/\eps_0)$. This allows us to apply Corollary~\ref{cor:cto1-opt}
  and get a $(1+\eps, \eps\gamma^{-2})$-solver with respect to $\|\cdot\|_1$
  with the running time claimed in the theorem.
  To get a truly feasible flow $f$,
  we compose this solver with the algorithm in
  Lemma~\ref{lm:preconditioner} which, for a supplies vector
  $\tilde{b}$, computes a feasible flow with cost bounded by $\gamma
  \|B\tilde{b}\|_1\le \gamma\cdot \cost(G,\tilde{b})$, and is,
  therefore, a $(\gamma, 0)$-solver. By Theorem~\ref{thm:composition},
  this composed algorithm is a $(1 + 2\eps, 0)$-solver. The
  total running time is dominated by the solver guaranteed by
  Corollary~\ref{cor:cto1-opt}.
\end{proof}

\section{Generating a Transportation Map}
\label{sect:flow2transp}


\newcommand\nfp{uniform flow parity\xspace}

Theorem~\ref{thm:flow} guarantees we can obtain an approximately
optimal flow in our graph $G$ in nearly linear time.  We wish to
turn this flow into a transportation map on $P$.  To accomplish this,
we will repeatedly transform our flow into other flows, without
increasing the cost, and ending at a flow which is also a
transportation map. We will no longer keep the flow supported on the
edges of $E$, and will instead allow positive flow between arbitrary
pairs of points. We will gradually make sure that there is positive
flow only between points in $P$ (as opposed to net points).

We first begin by defining a notion we call \nfp. 

\begin{definition}\label{def:nfp}
  Given a flow $f \in \R^{V \times V}$, and a vertex $u\in V$, $f$ is
  said to satisfy, or have \nfp at $u$ if for all remaining $v\in V$,
  $f_{uv}$ has the same sign.  In other words, either there is flow
  going out of $v$ or there is flow going into $v$, but there is no
  flow passing through $v$.
  The flow $f$ is said to satisfy \nfp if $f$ has \nfp on every
  $v\in V$.
\end{definition}

The next easy lemma shows that a flow that satisfies \nfp is supported
on edges between vertices with non-zero divergence.

\begin{lemma}\label{lm:nfp}
  Suppose that a flow $f \in \R^{V \times V}$ satisfies \nfp at a
  vertex $u \in V$, and that $\sum_{v \in V}{f_{uv}} = 0$.
  Then $f_{uv} = 0$ for all $v \in V$.
\end{lemma}
\begin{proof}
  Assume, towards contradiction, that $f$ has $0$ divergence at $u$
  but $f_{uv} \neq 0$ for some $v \in V$. Then there must be some $w
  \in V$ such that $f_{uw}$ has the opposite sign to $f_{uv}$, or
  $\sum_{v \in V}{f_{uv}} = 0$ would not hold. This contradicts \nfp.
\end{proof}




We apply Lemma~\ref{lm:nfp} through the following corollary of it.
\begin{corollary}\label{cor:nfp}
  Let $f$ be a flow on $V\times V$. Suppose that
  for any $u \in V \setminus P$, $\sum_{v \in V}{f_{uv}} = 0$ and for
  any $p \in P$, $\sum_{v \in V}{f_{pv}} = \mu(p)$. Then, if $f$
  satisfies \nfp, it is a transportation map for $P$ and $\mu$.
\end{corollary}
\begin{proof}
  By Lemma~\ref{lm:nfp}, $f$ is supported on $P \times P$. Moreover,
  since for any $p \in P_+$ we have $\sum_{q \in P}{f_{pq}} = \mu(p) >
  0$, and $f$ satisfies \nfp, it must be the case that $f_{pq} \ge 0$
  for all $q \in P$. Similarly, for any $p \in P_-$ it must be the
  case that $f_{pq} \le 0$ for all $q \in P$, or, equivalently,
  $f_{qp} \ge 0$ for all $q \in P$. It follows  that $f$ is in fact
  supported on $P_+ \times P_-$, and is, therefore, a transportation
  map. 
\end{proof}

Let us take a flow $f$ which is an approximately optimal solution to
\eqref{eq:mincost-obj}--\eqref{eq:mincost-fc}. We can extend $f$ to $V\times V$
by setting $f_{uv} = 0$ for $(u,v) \not \in E$. If we can
transform $f$ into another flow $f'$ without increasing its cost, so
that $f'$ satisfies \nfp, then we can use $f'$ as an approximately optimal
transportation map. 
Towards this goal, we define the Cancellation Procedure
\Call{Cancel-Vertex}{$f, u$}, given in Algorithm~\ref{alg:cancelv},
and analyzed in the following lemma. 

\begin{algorithm}[htp]
  \caption{The Cancellation Procedure}
  \label{alg:cancelv}
  \begin{algorithmic}[1]
    \Procedure{Cancel-Vertex}{$f, u$}
    \While{$\exists v, w: f_{vu} > 0 > f_{wu}$}
       \State $x = \min\{f_{vu}, f_{uw}\}$
       \State $f_{vu} \gets f_{vu} - x$
       \State $f_{uw} \gets f_{uw} - x$
       \State $f_{vw} \gets f_{vw} + x$
    \EndWhile
    \EndProcedure
  \end{algorithmic}
\end{algorithm}


\begin{lemma}\label{lm:cancelation}
  Let $f$ be a flow on $V \times V$, and $u \in V$. The while loop in
  \Call{Cancel-Vertex}{$f, u$} makes at most $|\{v: f_{uv} \neq 0\}|$
  many iterations, and, letting $f'$ be the flow after the procedure is
  called, we have the following properties:
  \begin{enumerate}
  \item All divergences are preserved, i.e.
    \[
    \sum_{w \in V}{f'_{vw}} = \sum_{w \in V}{f_{vw}}
    \]
    for all $v\in V$. 
  \item The flow   $f'$ satisfies \nfp at $u$. 
  \item If $f$ satisfies \nfp at some vertex $v \in V$, then so does
    $f'$.
  \item The size of the support of $f$ is at most that of $f'$.
  \item The cost of $f'$ with respect to the cost function $c(u,v) =
    \|u-v\|_2$ is at most the cost of $f$.
  \end{enumerate}
\end{lemma}
\begin{proof}
  To bound the number of iterations of the while loop, observe that
  every iteration sets the flow of at least one edge incident on $u$ to $0$,
  and that the flow on an edge incident on $u$ is modified only if it
  nonzero. We proceed to prove the other claims.

  The first claim follows because every iteration of the while loop of
  the procedure preserves the divergence $\sum_{w \in V}{f_{zw}}$ at every
  node $z \in V$: if $z = u$, then we remove $x$ units of flow going
  into $u$, and $x$ units of flow leaving $u$; otherwise, if $z = v$ or
  $z = w$ we just redirect $x$ units of flow entering or leaving $v$
  to another vertex; finally, if $z$ is any other vertex, then the
  flow through it is unaffected. 

  The fact that $f'$ satisfies \nfp on $u$ follows from the
  termination condition of the loop.

  To prove that \nfp is preserved, again we show that this is the case
  for any iteration of the while loop. For $u$, we observe that the
  the flow is unchanged if $u$ already satisfies \nfp. The only other
  vertices for which the flow on incident edges can be modified during
  an iteration are $v$ and $w$. For $v$, we have that $f_{vu} > 0$,
  so, if $v$ has \nfp, then the flow on every edge leaving $v$ must be
  non-negative. We decrease the value of $f_{vu}$ without making it
  negative (by the definition of $x$), and we increase the value of
  $f_{vw}$, which was non-negative to begin with. Therefore, \nfp is
  preserved for $v$. The argument for $w$ is analogous.

  To see that that the number of edges in the support of $f$
  does not increase, notice that every iteration of the while loop
  removes at least one edge from the support of $f$, and adds flow to
  at most one edge which previously had flow of $0$, namely $(v,w)$. 

  Finally, to show that the cost of the flow does not increase,
  observe that every iteration of the while loop replaces $x$ units of
  flow of total cost $x(\|v-u\|_2 + \|w - u\|_2)$ by $x$  units of
  flow of total cost $x \|w-v\|_2$, which is no larger than the
  original cost by the triangle inequality. 
\end{proof}

Suppose that we maintain $f$ in a sparse representation; namely, we
keep an adjacency list of the edges on which $f$ is not zero, with the
corresponding flow values. Then, after preprocessing the adjacency
list of $u$ in linear time to find the edges with positive and
negative flow,  every iteration of the while loop in
the Cancellation Procedure can be executed in constant time, and the
total running time, by Lemma~\ref{lm:cancelation}, is bounded by
$O(\deg_{f}(u))$, where $\deg_{f}(u)=|\{v: f_{uv} \neq 0\}|$.




Corollary~\ref{cor:nfp} and Lemma~\ref{lm:cancelation} imply that if
we apply the Cancellation Procedure to the flow $f$ and every vertex
of the graph $G$, then the resulting flow has cost no greater 
than that of $f$, is supported on a set of edges of size bounded by $|E|$, and is a transportation map.  In the
next theorem, which is our main result for constructing a
transportation map from a flow $f$ in $G$, we show that
if we apply the procedure first to netpoints in $N_L$, then to
$N_{L-1}$, etc., then the total running time
is nearly linear.

\begin{theorem}\label{thm:flow2transp}
  There exists an algorithm running in time 
  $O\left({n\eps^{-d}\log(\Delta)} + {\eps^{-2d}}{\log(\Delta)}\right)$ that,
  given as input a flow $f$ which is feasible for the minimum cost
  flow problem defined in Section~\ref{sect:min-cost-flow}, outputs a
  transportation map for $P$ and $\mu$ of cost no larger than that of
  $f$.
\end{theorem}
\begin{proof}
  As discussed above, we (mentally) extend $f$ to all of $V \times V$ by defining
  $f_{uv} = 0$ if $(u,v) \not \in E$. 
  We  call \Call{Cancel-Vertex}{$f, u$} on $f$ and every net point $u
  \in N_L$ in arbitrary order; then we call it on every net point $u
  \in N_{L-1}$, and so on, finally calling it on every net point in
  $N_0$. By Lemma~\ref{lm:cancelation}, after these calls \nfp must
  hold at all net points. Moreover, it also holds at points
  in $P$, since in $G$ they only have a single edge incident on them,
  so initially \nfp holds trivially, and by
  Lemma~\ref{lm:cancelation} it is not destroyed by later calls to the
  Cancellation Procedure. Thus, the flow at termination has \nfp at all
  vertices, and, by Corollary~\ref{cor:nfp}, is a transportation
  map. Since calls to the Cancellation Procedure do not increase the
  cost of $f$, the cost of the transportation map is at most the cost
  of the initial flow. 
  
  Let us then consider a time step when the Cancellation
  Procedure has been applied to all net points in $N_{\ell+1} \cup
  \ldots \cup N_L$ and some points of $N_\ell$, but no points in $N_0
  \cup \ldots \cup N_{\ell-1}$. By Lemma~\ref{lm:nfp}, the flow $f$ at
  such a step is supported on edges incident on $N_0 \cup \ldots \cup N_\ell \cup
  P$. Recall that in $G$ net points on level $\ell$ only
  connect to other net points on the same level and in the same cell,
  and to parent net points on level $\ell-1$. Then
  we have that for any $\cell \in
  \grid_\ell$ and the cell $\cell' \in \grid_{\ell-1}$ containing it,
  and for any point $u \in N_\ell \cap \cell$, 
  $f_{uv}$ can be nonzero only for $v \in (N_\ell \cap \cell) \cup (P \cap
  \cell) \cup (N_{\ell-1} \cap \cell')\}$. Let
  $n_{\text{below}}(\cell)=\abs{P\cap\cell}$, and recall the notation
  $\deg_f(u) = |\{v \in V: f_{uv} \neq 0\}|$. We then have that
  $\deg_f(u) \le n_{\text{below}}(\cell) + {2}{\eps^{-d}}$, which,
  by Lemma~\ref{lm:cancelation}, also dominates the running time of
  the Cancellation Procedure on $u$. The total complexity of running
  the Cancellation Procedure on all points in $N_\ell$ is then
  dominated by
  \[
  \sum_{\cell \in \grid_\ell}{\frac{1}{\eps^d}(n_{\text{below}}(\cell) +
    \frac{2}{\eps^d})} 
  = O(\eps^{-d} n + \eps^{-2d}). 
  \]
  Summing the running times over the $L + 1 = O(\log \Delta)$ levels
  finishes the proof.
\end{proof}

Combining Theorem~\ref{thm:mincost-apx}, used with $\eps_0$ set to a
sufficiently small multiple of $\eps/L$, and
Theorems~\ref{thm:flow}~and~\ref{thm:flow2transp} gives the first
claim of Theorem~\ref{thm:main}, but with the approximation holding in
expectation. At the cost of increasing $\eps$ by a factor of $2$,
Markov's inequality shows that the approximation also holds with
probability at least $1/2$. Then the second statement in
Theorem~\ref{thm:main} follows from the first one, and
Lemma~\ref{lm:reduce-aspect}.






\bibliographystyle{alpha}
\bibliography{emd}

\newcommand{\etalchar}[1]{$^{#1}$}
\begin{thebibliography}{BvdPPH11}

\bibitem[ABRW18]{Altschuler-W2}
Jason Altschuler, Francis Bach, Alessandro Rudi, and Jonathan Weed.
\newblock Approximating the quadratic transportation metric in near-linear
  time.
\newblock {\em CoRR}, abs/1810.10046, 2018.

\bibitem[AFP{\etalchar{+}}17]{AgarwalFPVX17}
Pankaj~K. Agarwal, Kyle Fox, Debmalya Panigrahi, Kasturi~R. Varadarajan, and
  Allen Xiao.
\newblock Faster algorithms for the geometric transportation problem.
\newblock In {\em Symposium on Computational Geometry}, volume~77 of {\em
  LIPIcs}, pages 7:1--7:16. Schloss Dagstuhl - Leibniz-Zentrum fuer Informatik,
  2017.

\bibitem[AHK12]{AroraHK12}
Sanjeev Arora, Elad Hazan, and Satyen Kale.
\newblock The multiplicative weights update method: a meta-algorithm and
  applications.
\newblock {\em Theory of Computing}, 8(1):121--164, 2012.

\bibitem[ANOY14]{AndoniNOY14}
Alexandr Andoni, Aleksandar Nikolov, Krzysztof Onak, and Grigory Yaroslavtsev.
\newblock Parallel algorithms for geometric graph problems.
\newblock In {\em {STOC}}, pages 574--583. {ACM}, 2014.

\bibitem[Aro98]{Arora98}
Sanjeev Arora.
\newblock Polynomial time approximation schemes for euclidean traveling
  salesman and other geometric problems.
\newblock {\em J. {ACM}}, 45(5):753--782, 1998.

\bibitem[AS14]{AgarwalS14}
Pankaj~K. Agarwal and R.~Sharathkumar.
\newblock Approximation algorithms for bipartite matching with metric and
  geometric costs.
\newblock In {\em {STOC}}, pages 555--564. {ACM}, 2014.

\bibitem[BvdPPH11]{Bonneel11}
Nicolas Bonneel, Michiel van~de Panne, Sylvain Paris, and Wolfgang Heidrich.
\newblock Displacement interpolation using lagrangian mass transport.
\newblock {\em ACM Trans. Graph.}, 30(6):158:1--158:12, December 2011.

\bibitem[Cha02]{Charikar02}
Moses Charikar.
\newblock Similarity estimation techniques from rounding algorithms.
\newblock In {\em {STOC}}, pages 380--388. {ACM}, 2002.

\bibitem[CK93]{CallahanK93}
Paul~B. Callahan and S.~Rao Kosaraju.
\newblock Faster algorithms for some geometric graph problems in higher
  dimensions.
\newblock In {\em {SODA}}, pages 291--300. {ACM/SIAM}, 1993.

\bibitem[FKS84]{FKS84}
Michael~L. Fredman, J\'{a}nos Koml\'{o}s, and Endre Szemer\'{e}di.
\newblock Storing a sparse table with {$O(1)$} worst case access time.
\newblock {\em J. Assoc. Comput. Mach.}, 31(3):538--544, 1984.

\bibitem[GD04]{GraumanD04}
Kristen Grauman and Trevor Darrell.
\newblock Fast contour matching using approximate earth mover's distance.
\newblock In {\em Computer Vision and Pattern Recognition, 2004. CVPR 2004.
  Proceedings of the 2004 IEEE Computer Society Conference on}, volume~1, pages
  I--I. IEEE, 2004.

\bibitem[GV02]{GiannV02}
Panos Giannopoulos and Remco~C Veltkamp.
\newblock A pseudo-metric for weighted point sets.
\newblock In {\em European Conference on Computer Vision}, pages 715--730.
  Springer, 2002.

\bibitem[HP11]{Sariel-book}
Sariel Har-Peled.
\newblock {\em Geometric approximation algorithms}, volume 173 of {\em
  Mathematical Surveys and Monographs}.
\newblock American Mathematical Society, Providence, RI, 2011.

\bibitem[Ind07]{Indyk07}
Piotr Indyk.
\newblock A near linear time constant factor approximation for euclidean
  bichromatic matching (cost).
\newblock In {\em {SODA}}, pages 39--42. {SIAM}, 2007.

\bibitem[IT03]{IndykThaper03}
Piotr Indyk and Nitin Thaper.
\newblock {Fast Image Retrieval via Embeddings}.
\newblock In {\em 3rd International Workshop on Statistical and Computational
  Theories of Vision}, 2003.

\bibitem[LCL04]{QinCL04}
Qin Lv, Moses Charikar, and Kai Li.
\newblock Image similarity search with compact data structures.
\newblock In {\em Proceedings of the Thirteenth ACM International Conference on
  Information and Knowledge Management}, CIKM '04, pages 208--217, New York,
  NY, USA, 2004. ACM.

\bibitem[PA06]{PhillipsA06}
Jeff~M. Phillips and Pankaj~K. Agarwal.
\newblock On bipartite matching under the {RMS} distance.
\newblock In {\em {CCCG}}, 2006.

\bibitem[RTG00]{RubnerTG00}
Yossi Rubner, Carlo Tomasi, and Leonidas~J Guibas.
\newblock The earth mover's distance as a metric for image retrieval.
\newblock {\em International journal of computer vision}, 40(2):99--121, 2000.

\bibitem[SA12a]{SharathkumarA12-SODA}
R.~Sharathkumar and Pankaj~K. Agarwal.
\newblock Algorithms for the transportation problem in geometric settings.
\newblock In {\em {SODA}}, pages 306--317. {SIAM}, 2012.

\bibitem[SA12b]{SharathkumarA12-STOC}
R.~Sharathkumar and Pankaj~K. Agarwal.
\newblock A near-linear time {\(\epsilon\)}-approximation algorithm for
  geometric bipartite matching.
\newblock In {\em {STOC}}, pages 385--394. {ACM}, 2012.

\bibitem[Sch03]{Schrijver-volA}
Alexander Schrijver.
\newblock {\em Combinatorial optimization. {P}olyhedra and efficiency. {V}ol.
  {A}}, volume~24 of {\em Algorithms and Combinatorics}.
\newblock Springer-Verlag, Berlin, 2003.
\newblock Paths, flows, matchings, Chapters 1--38.

\bibitem[She17]{Sherman17}
Jonah Sherman.
\newblock Generalized preconditioning and undirected minimum-cost flow.
\newblock In {\em {SODA}}, pages 772--780. {SIAM}, 2017.

\bibitem[Vil03]{Villani-Topics}
C\'{e}dric Villani.
\newblock {\em Topics in optimal transportation}, volume~58 of {\em Graduate
  Studies in Mathematics}.
\newblock American Mathematical Society, Providence, RI, 2003.

\end{thebibliography}

\appendix
\section{Proof of Lemma~\ref{lm:distort}}

  The first claim follows from the triangle inequality, since
  $\dist_G(p,q)$ is the length of a polygonal path between $p$ and
  $q$, and $\|p-q\|_2$ is the length of the straightline segment
  between them. 

  For the second claim, it is enough to show that the canonical path
  between $p$ and $q$ has the required length. Observe first that, if
  $u = N_{\ell-1}(v)$ is the parent net point of some $v\in N_\ell$,
  then $\|u-v\|_2=\frac{\sqrt{d}\eps_0\Delta_\ell}{2}$.  Similarly, if
  $r\in P$, then
  $\|r-N_\ell(r)\|_2\le\frac{\sqrt{d}\eps_0\Delta_\ell}{2}$.  This
  shows that
  \begin{align*}
    \|N_{\ell(p,q)}(p)-N_{\ell(p,q)}(q)\|_2&\leq\|N_{\ell(p,q)}(p)-p\|_2
    +\|p-q\|_2+\|q-N_{\ell(p,q)}(q)\|\\
    &\le\|p-q\|_2+\sqrt{d}\eps_0\Delta_{\ell(p,q)}.
  \end{align*}

  Using the canonical path to bound the shortest path distance between
  $p$ and $q$ in $G$, we see that 
\begin{align*}
\dist_G(p,q)&\leq\|p-N_L(p)\|_2+\sum\limits_{\ell=\ell(p,q)+1}^L
\|N_\ell(p)-N_{\ell-1}(p)\|_2\\
&+\|N_{\ell(p,q)}(p)-N_{\ell(p,q)}(q)\|_2+\sum\limits_{\ell=\ell(p,q)+1}^L
\|N_{\ell-1}(q)-N_\ell(q)\|_2+\|N_L(q)-q\|_2\\
&\le\|p-q\|_2+\sqrt{d}\eps_0\left(\Delta_{\ell(p,q)}+\Delta_L+
\sum\limits_{\ell=\ell(p,q)+1}^L\Delta_\ell\right)\\
&=\|p-q\|_2+\sqrt{d}\eps_0\left(2^{1-\ell(p,q)}\Delta+2^{1-L}\Delta+
\sum\limits_{\ell=\ell(p,q)+1}^L2^{1-\ell}\Delta\right)\\
&\leq\|p-q\|_2+\sqrt{d}\eps_0\Delta(2^{1-\ell(p,q)}+2^{1-L}+
2^{1-\ell(p,q)})\\
&\leq\|p-q\|_2+3\sqrt{d}\eps_0\Delta_{\ell(p,q)}
\end{align*}

This shows that, conditional on ${\ell(p,q)}$,
$\dist_G(p,q)=\|p-q\|_2+O(\eps_0\Delta_{\ell(p,q)})$.
All that remains is to find the expected value of $\dist_G(p,q)$
by taking expectations over the 
distribution of ${\ell(p,q)}$.
We get that the expected value of $\dist_G(p,q)$ is
\begin{align*}
\E[\dist_G(p,q)]&\le \|p-q\|_2+3\sqrt{d}\eps_0\E[\Delta_{\ell(p,q)}]\\
&=\|p-q\|_2+3\sqrt{d}\eps_0\sum\limits_{\ell=0}^{L-1}\Pr({\ell(p,q)}=\ell)
\Delta_{\ell}\\
&=\|p-q\|_2+3\sqrt{d}\eps_0\sum\limits_{\ell=0}^{L-1}\Pr(\cell_\ell(p)=\cell_\ell(q)
\wedge\cell_{\ell+1}(p)\neq\cell_{\ell+1}(q))\Delta_{\ell}\\
&\leq\|p-q\|_2+3\sqrt{d}\eps_0\sum\limits_{\ell=0}^{L-1}
\Pr(\cell_{\ell+1}(p)\neq\cell_{\ell+1}(q))\Delta_{\ell}\\
&=\|p-q\|_2+3\sqrt{d}\eps_0\sum\limits_{\ell=1}^{L}
\Pr(\cell_{\ell}(p)\neq\cell_{\ell}(q))\Delta_{\ell}\\
&\le\|p-q\|_2+3\sqrt{d}\eps_0
\sum\limits_{\ell=1}^L\frac{\sqrt{d}\|p-q\|_2}{\Delta_\ell}\Delta_{\ell}\\
&=(1+3{d}\eps_0L)\|p-q\|_2.
\end{align*}
The final inequality above follows from Lemma~\ref{lm:sep-prob}.
This completes the proof of the lemma.

\section{Sherman's Algorithms}

To make our paper more self-contained, in this section we briefly
sketch a version of the algorithm from Theorem~\ref{thm:1to1-opt}
which uses the multiplicative weight update method~\cite{AroraHK12} in
a simple way, but achieves a slightly worse running time than the one
stated in the theorem. We do not claim any novelty, and everything in
this section is due to Sherman.

One of Sherman's main results is an analysis of how composing rough
solvers affects the error in the linear constraints. In particular,
suppose that $\mathcal{F}$ is a $(1+\eps,
\frac{\eps}{2\kappa})$-solver, $\mathcal{G}$ is $(2,
\frac{1}{2\kappa})$-solver, and $\mathcal{G}^t$ is $\mathcal{G}$
composed with itself $t$ times. (The notion of composition here is the
same one used in Theorem~\ref{thm:composition}.) Then, the composition
$\mathcal{G}^t \circ \mathcal{F}$ is a $(1 + O(\eps), \frac{\eps
  2^{-t-1}}{\kappa})$-solver. This reduces designing the solver from
Theorem~\ref{thm:1to1-opt} to designing a $(1+\eps, \eps)$-solver
that performs $O(\eps^{-2} \log(m))$ multiplications with $A$ or
$A^\top$ (and possibly lower-order operations). We will sketch how to
get a solver whose complexity is instead
$O(\eps^{-2}\log(m)\log_{1+\eps}(\kappa))$ multiplications with $A$ or
$A^\top$.

By scaling, we can assume that $\|A\|_{\|\cdot\|_1 \to \|\cdot\|_1}
\le 1$, i.e.~every column of the matrix $A$ has $\ell_1$-norm bounded
by $1$. It then follows that for any $f$ that satisfies the constraint
$Af = b$, we have $\|f\|_1 \ge \|b\|_1$, i.e.~the optimal value of the
optimization problem \eqref{eq:minnorm-obj}--\eqref{eq:minnorm-constr}
is at least $\|b\|_1$. Moreover, if the condition number of $A$ is
$\kappa$, then there exists a map $S$ from the span of $A$, equipped
with $\ell_1$, to $\R^m$, equipped with $\ell_1$, such that $\|S\|\le
\kappa$ and $AS(b) = b$. It follows that $\|S(b)\|_1 \le \kappa
\|b\|_1$, so the optimal value lies in $[\|b\|_1, \kappa \|b\|_1]$. We
can then guess a value in this range (up to a factor of $1+\eps$), and
reduce the problem to making $O(\log_{1+\eps} \kappa)$ calls to an
algorithm solving the feasibility problem $\{f: \|f\|_1 \le t, Ax =
b\} \neq \emptyset$. This is the feasibility problem that we will
approximately solve using the multiplicative weights framework.

In fact, it will be convenient to reformulate the above feasibility
problem as $\{f: \|f\|_1 \le 1, Ax = b/t\} \neq \emptyset$. In order
to get a $(1+\eps, O(\eps))$-solver, it is enough to output an $f$
such that $\|f\|_1 \le 1$ and $\|Af - \frac1t b\|_1 \le \eps$ whenever
the problem is feasible. We will write $f = f^+ - f^-$, where both
$f^+$ and $f^-$ have non-negative coordinates and $\sum_{i =
  1}^m{f^+_i} = \sum_{i = 1}^m{f^-_i} = 1$. We initialize $f^+$ and
$f^-$ to $f^+_i = f^-_i = \frac1m$ for all $i$, and we update them
iteratively using the multiplicative weights update method. In each
step of the algorithm, we either have $\|Af - \frac1t b\|_1 \le \eps$ and we
are done, or we can find a vector $y \in \{-1, 1\}^n$ such that
$y^\top (Af - b) > \eps$. We update
\begin{align*}
f^+_i &= \frac{e^{-\eta (A^\top y)_i}}{\sum_{j = 1}^me^{-\eta (A^\top  y)_j}},
&f^-_i = \frac{e^{\eta (A^\top y)_i}}{\sum_{j = 1}^me^{\eta (A^\top y)_j}},
\end{align*}
where $\eta$ is a parameter. If we choose $\eta = \frac{\eps}{2}$, a
standard analysis of the multiplicative weights
method~\cite{AroraHK12} shows that, after $O(\eps^{-2} \log m)$
iterations, we would either have $\|A(f^+ - f^-) - \frac1t b\|_1 \le
\eps$, as desired, or the average $\bar{y}$ of all vectors $y$ computed so far
satisfies $\pm\bar{y}^\top a_i > \bar{y}^\top b$ for every column $a_i$
of $A$, certifying that the problem is infeasible. 

\end{document}